\documentclass[a4paper,conference]{IEEEtran}
\usepackage{tikz}
\usetikzlibrary{decorations.pathreplacing}
\usetikzlibrary{shapes.misc, positioning}
\usepackage{microtype}
\usepackage{relsize}
\usepackage{setspace}
\usepackage{algorithm}
\usepackage{algorithmicx}
\usepackage{algpseudocode}
\usepackage{epsfig}  
\usepackage{graphicx}  
\usepackage{longtable}
\usepackage{amsmath}  
\usepackage{array}
\usepackage{multirow}
\usepackage{rotating}
\usepackage{amssymb}  
\usepackage{lscape}  
\usepackage{amsthm}
\usepackage[justification=raggedright]{caption}	
\newtheorem{theorem}{Theorem}

\newtheorem{lemma}{Lemma}
\newtheorem{example}{Example}

\ifCLASSINFOpdf
\else
\fi
\begin{document}
%
\title{An Achievable Rate Region for the Two-Way Multiple Relay Channel}
%
%
%

\author{\IEEEauthorblockN{Jonathan Ponniah~\IEEEmembership{Member,~IEEE,}
}
\IEEEauthorblockA{Department of Electrical and Computer Engineering\\
 Texas A\&M University
}
\and
\IEEEauthorblockN{Liang-Liang Xie~\IEEEmembership{Fellow,~IEEE,}
}
\IEEEauthorblockA{Department of Electrical and Computer Engineering\\
University of Waterloo
}\thanks{This material is based upon work partially supported by NSF Contract CNS-1302182,  AFOSR Contract FA9550-13-1-0008, and NSF Science \& Technology Center Grant CCF-0939370.}}

%



\maketitle

\begin{abstract}
An achievable rate-region for the two-way multiple-relay channel is proposed using decode-forward block Markovian coding.  We identify a fundamental tension between the information flow in both directions that leads to an intractable number of decode-forward schemes and achievable rate regions, none of which are universally better than the others.  We introduce a new concept in decode-forward coding called ranking, and discover that each of these rate regions are different realizations of a single expression that depends on the rank assignment.  This discovery makes it possible to characterize the complete achievable rate region that includes all of the interesting decode-forward schemes and corresponding rate regions.  
\end{abstract}


%
\IEEEpeerreviewmaketitle
\section{Introduction}

The two-way multiple relay channel (TWMRC) is implicit in almost all forms of modern communication.  At the most fundamental level, it models the simultaneous, interactive, bidirectional exchange of information between two source/destination pairs with the assistance of intermediate relay nodes; the type of exchange that occurs over the internet for instance, where every single received data packet is acknowledged.  We consider the most general form of the TWMRC which allows the transmissions of every node to influence the received signal of each individual node.

This paper extends a history of previous work in network information theory beginning with the two-way channel \cite{Shannon}, followed by the one-way relay channel \cite{ElGamal}, the one-way multiple-relay channel \cite{Xie2005}, the one-way multiple-access relay channel \cite{Sankar}, the two-way, one-relay channel \cite{Xie2007}, and the two-way, two-relay channel \cite{Ponniah}.  

The channel model of interest is the discrete memoryless channel consisting of $M$ nodes labeled from $1,\ldots,M$.  The input-output dynamics are expressed as follows: 
\begin{align}
	\nonumber
	({\cal X}_{1}\times\ldots\times{\cal X}_{M},p(y_{1},\ldots,y_{M}|x_{1},\ldots,x_{M}),{\cal Y}_{1}\times\ldots\times{\cal Y}_{M}).
\end{align}
That is, at every time instant $t=1,2,\ldots$, the outputs $y_{1}(t),\ldots,y_{M}(t)$ received by the $M$ nodes respectively only depend on the inputs $x_{1}(t),\ldots,x_{M}(t)$ transmitted by the $M$ nodes at the same time according to $p(y_{1}(t),\ldots,y_{M}(t)|x_{1}(t),\ldots,x_{M}(t))$.  If node $i$ is a relay, the input $x_{i}(t)$ into the channel at time $t$ depends only on the symbols received in the previous time instants, so that $x_{i}(t)=F_{i,t}(y_{i}(t-1),y_{i}(t-2),\ldots)$ for all $t$, where $F_{i,t}$ can be any causal function.  

Consider the multi-source, multi-relay, all-way channel in which each source is interested in the messages generated by all the other sources (the TWMRC is a special case of this channel).  In a \textit{decode-forward scheme} the messages generated by each source are successively decoded and forwarded by each relay in the channel before arriving at any destination.  The order in which the nodes forward the message from a particular source defines a \textit{path}.  Let ${\cal S}$ denote the set of source nodes.  For each $s\in{\cal S}$, let the vector ${\bar p}_{s}$ denote the fixed path assigned to source $s$ and let ${\bar p}:=\{{\bar p}_{s}: s\in{\cal S}\}$.  The first element of ${\bar p}_{s}$ is source node $s$, and each subsequent element is the next node on the path.  Let $p_{s,k}$ denote the $k$th element of ${\bar p}_{s}$.   
\begin{example}
	Set $M:=5$ and ${\cal S}:=\{1,5\}$.  Suppose the messages $w_{1}$ and $w_{5}$ follow the paths $1\rightarrow3\rightarrow2\rightarrow4\rightarrow5$ and $5\rightarrow2\rightarrow3\rightarrow4\rightarrow1$ respectively.  Then ${\bar p}_{1}=(1,3,2,4,5)$, ${\bar p}_{5}=(5,2,3,4,1)$, $p_{1,2}=3$, and ${\bar p}:=\{{\bar p}_{1},{\bar p}_{5}\}$. 
\end{example}

For any fixed ${\bar p}$, let $P_{i}(s)$ denote the set of nodes that precede node $i$ on the path ${\bar p}_{s}$.  Suppose node $i$ is the $k^{th}$ node on the path ${\bar p}_{s}$.  That is, $i=p_{s,k}$.  Then $P_{i}(s):=\{p_{s,j}: 1\leq j\leq k-1\}$.  For any non-empty subset $S\subseteq{\cal S}$, let $P_{i}(S)=\displaystyle\bigcup_{s\in S} P_{i}(s)$ and $\tilde{P}_{i}(S):=\{k\notin P_{i}(S)\}$.
\begin{example}
	Set $M:=5$ and ${\cal S}:=\{1,5\}$. Set ${\bar p}_{1}:=(1,3,2,4,5)$ and ${\bar p}_{5}:=(5,2,3,4,1)$.   Then $P_{2}(\{1\})=\{1,3\}$, $P_{2}(\{5\})=\{5\}$, and $P_{2}(\{1,5\})=\{1,3,5\}$.  Similarly, $\tilde{P}_{2}(\{1\})=\{2,4,5\}$, $\tilde{P}_{2}(\{5\})=\{1,2,3,4\}$, and $\tilde{P}_{2}(\{1,5\})=\{2,4\}$.
\end{example}

Set $|{\cal S}|:=N$ for some $N\leq M$, and ${\cal S}:=\{s_{1},\ldots,s_{N}\}$, where $s_{j}\in\{1,\ldots,M\}$ for each $j=1,\ldots,N$.  For any ${\bar p}$ and non-empty $S\subseteq{\cal S}$ define the constraint:  
\begin{align}
	\label{decodeforwardupperbound}
	R_{S}<I(X_{P_{i}(S)};Y_{i}|X_{\tilde{P}_{i}(S)}),
\end{align}
where $R_{S}$ denotes the sum $\sum_{j\in S}R_{j}$, and $X_{S}$ denotes the set $\{X_{j}: j\in S\}$.   Fix ${\bar p}$ and let ${\cal R}_{i}({\bar p})$ be the set of rate vectors ${\bar R}:=(R_{1},\ldots,R_{N})$ that satisfy (\ref{decodeforwardupperbound}) for all non-empty $S\subseteq{\cal S}$ if $i\notin{\cal S}$ and all non-empty $S\subseteq{\cal S}\setminus i$ if $i\in{\cal S}$.

A rate vector ${\bar R}:=(R_{1},\ldots,R_{N})$ for a multi-source, multi-relay, all-way channel with $N$ sources is \textit{achievable} by definition, if there exists an encoding/decoding scheme that allows source node $s_{j}$ for each $j\in\{1,\ldots,N\}$ to send information at rate $R_{j}$ to all the other sources with an arbitrarily small probability of error.  An outer-bound on the region of achievable rate vectors that can be recovered using decode-forward schemes in this channel is given by: 
\begin{align}
	\label{decodeforwardrateregion}
	{\cal C}_{d}:=\displaystyle\bigcup_{{\bar p}}\displaystyle\bigcap^{M}_{i=1}{\cal R}_{i}({\bar p}).
\end{align}

To observe the difference between (\ref{decodeforwardrateregion}) and the cut-set outer-bound, replace (\ref{decodeforwardupperbound}) with the following constraint:
\begin{align}
	\label{cutset}
	R_{S}<I(X_{P_{i}(S)};Y_{\tilde{P}_{i}(S)}|X_{\tilde{P}_{i}(S)}).
\end{align}
Assume multi-way communication so that $N\geq2$.  For any ${\bar p}$, let $\hat{{\cal R}}_{i}({\bar p})$ be the set of rate vectors ${\bar R}:=(R_{1},\ldots,R_{N})$ that satisfy (\ref{cutset}) for all non-empty $S\subset{\cal S}$ if $i\notin{\cal S}$ and all non-empty $S\subseteq{\cal S}\setminus i$ if $i\in{\cal S}$.  A key difference between $\hat{{\cal R}}_{i}({\bar p})$ and ${\cal R}_{i}({\bar p})$ is that the former requires $S$ to be a strict subset of ${\cal S}$ for $N\geq2$ since there is no cut that puts all the sources on the same side in multi-way communication.  The cut-set outer-bound corresponds to the following region:
\begin{align}
	\label{cutsetupperbound}
	{\cal C}:=\displaystyle\bigcap_{{\bar p}}\displaystyle\bigcap^{M}_{i=1}\hat{{\cal R}}_{i}({\bar p}).
\end{align}

The regions (\ref{cutsetupperbound}) and (\ref{decodeforwardrateregion}) coincide if $N=1$ and there is one path ${\bar p}$ over which the channel is physically degraded.  In the one-way multi-relay channel (OWMRC), the two-way one-relay channel, and the three-way broadcast channel \cite{Xie2007},  ${\cal C}_{d}$ can be achieved for all joint-distributions in the first case and all product distributions in the latter two cases.  The key feature of the channel that is exploited in \cite{Xie2005} is the unidirectional flow of information. Node $i$ can remove the interference generated by the nodes in $\tilde{P}_{i}(S)$ because these nodes being downstream of node $i$, transmit messages already decoded by $i$.  

It turns out that any attempt to recover the rate vectors in ${\cal C}_{d}$ for multi-way channels with two or more relays encounters a fundamental tension between the information flow in one direction and the information flow in the opposite direction.  This tension is illustrated in the two-way two-relay channel.  

\begin{example}
	The two-way two-relay channel consists of nodes $\{1,2,3,4\}$.  Define the set of source nodes as ${\cal S}:=\{1,4\}$ and consider the paths ${\bar p}_{1}:=(1,2,3,4)$ and ${\bar p}_{4}:=(4,3,2,1)$.  In order to decode a message $w_{1}$ from node 1 at the rate $R_{1}<I(X_{1};Y_{2}|X_{2}X_{3}X_{4})$, node 2 needs to know the message $w_{4}$ simultaneously transmitted by node 4.  But $w_{4}$ is new information and node 2 does not know it a priori.  Hence we have the following requirement:  
\begin{itemize}
	\item[(i)] Node 3 must first decode and forward $w_{4}$ before node 2 decodes $w_{1}$.  
\end{itemize}
However, the reverse situation occurs when node 3 tries to decode $w_{4}$ from node 4 at rate $R_{4}<I(X_{4};Y_{3}|X_{1}X_{2}X_{3})$.  Then we have the following requirement
\begin{itemize}
	\item[(ii)] Node 2 must first decode and forward $w_{1}$ before node 3 decodes $w_{4}$.
\end{itemize}
It is impossible to simultaneously satisfy (i) and (ii); either (ii) is satisfied at the expense of (i) or (ii) is satisfied at the expense of (i). 
\end{example}

Any attempt to recover ${\cal C}_{d}$ leads to a decode-forward scheme that decides, at each relay, the extent to which the OWMRC is simulated in one direction at the expense of the other direction.  Each decision prevents the decode-forward scheme from recovering some of the rate pairs in ${\cal C}_{d}$.  This tension generates many different decode-forward schemes and rate regions, all of which cumulatively fail to recover ${\cal C}_{d}$ and none of which include the others.  Furthermore, the regions recovered by each of these decode-forward schemes share no obvious pattern.  As a result, it becomes intractable to explicitly characterize the rate region that includes all possible decode-forward schemes for an arbitrary number of nodes.

The main contribution of this paper is the discovery that the rate region corresponding to any decode-forward scheme that attempts to recover ${\cal C}_{d}$ in the TWMRC is a particular realization of a single expression.  This expression depends on the \textit{rank assignment}, where the rank assigned to each node is determined by the decode-forward scheme that attempts to recover ${\cal C}_{d}$.  This property makes it possible to characterize all the interesting decode-forward schemes by describing the set of rankings instead.  

This paper focuses on the decode-forward relay scheme. Another important but fundamentally different relay scheme originally proposed in \cite{ElGamal} is the compress-forward scheme, which has also been successfully extended to more general networks in \cite{NoisyNetworkCoding} and \cite{WuXie}. It is well known that neither decode-forward nor compress-forward is absolutely better than the other, and their relative superiority depends on the network topology in general \cite{Schein}. However, for the two-way traffic considered in the paper, especially when relay nodes are evenly placed in between, it is arguably clear that decode-forward performs better.

The rest of this paper is organized as follows:  Section \ref{RankingSection} describes all valid rankings, Section \ref{mainresult} states the main result which is the complete achievable rate region, Section \ref{sectionproof} proves the main result and section \ref{conclusion} concludes the paper.

\section{Ranking}
\label{RankingSection}
Given a multi-source multi-relay all-way channel with $M$ nodes, a \textit{rank} index is a number between $1\ldots M$ uniquely assigned to each node.  The nodes are also labeled from $1,\ldots,M$ but the rank indices are distinct from the labels.  If node $i$ is assigned rank $k$, then ${\rm rank}(i)=k$.  A rank assignment $\bar{r}$ is a one-to-one mapping of rank indices to labels and is represented by an $M$-dimensional vector where ${\bar r}=(r_{1},r_{2},\ldots,r_{M})$ and $r_{i}={\rm rank}(i)$.  The rank indices are ordered by the binary relations ``$>$'', ``$<$'', ``$=$'', ``$\geq$'', and ``$\leq$''.  These relations retain their usual meaning in the sense that $M>M-1>\cdots>2>1$ and $1<2<\cdots<M-1<M$.  There is no ordering defined on the labels.  

The analysis in the sequel will be limited to the TWMRC with ${\cal S}:=\{1,M\}$ and path vector ${\bar p}:=\{{\bar p}_{1},{\bar p}_{M}\}$ where ${\bar p}_{1}:=(1,2,\ldots,M-1,M)$ and ${\bar p}_{M}:=(M,M-1,\ldots,2,1)$.  A path-rank-assignment pair $({\bar p},{\bar r})$ is \textit{valid} by definition if there is only one local minimum (with respect to the rank indices) over both paths ${\bar p}_{1}$ and ${\bar p}_{M}$.  More precisely, for each $s\in{\cal S}$, let $i_{s}:=\displaystyle\arg\min_{k}{\rm rank}(p_{s,k})$.  That is, the $i_{s}^{th}$ element of ${\bar p}_{s}$ has the lowest rank.  Then $({\bar p},{\bar r})$ is valid by definition if it satisfies the following conditions for all $s\in {\cal S}$: ${\rm rank}(p_{s,i})>{\rm rank}(p_{s,i+1})$ for all $i<i_{s}$ and ${\rm rank}(p_{s,i})>{\rm rank}(p_{s,i-1})$ for all $i>i_{s}$.  The notation ${\bar v}=({\bar p},{\bar r})$ will be used to denote a valid path-rank-assignment pair $({\bar p},\bar{r})$.  

\begin{example}
	The following are examples of rank assignments ${\bar r}$ that correspond to a valid $({\bar p},{\bar r})$ when $M=8$: {\rm (1,2,3,4,5,6,7,8)}, {\rm (8,7,6,5,4,3,2,1)}, {\rm (8,6,4,2,1,3,5,7)}, {\rm (8,7,6,4,3,1,2,5)}, {\rm (7,6,4,3,1,2,5,8)}.  
\end{example}

Let ${\cal V}$ denote the set of all valid path-rank-assignment pairs.  For a fixed ${\bar v}\in{\cal V}$ and each $s\in{\cal S}$, let $u(i,s)$ denote the ``one-hop'' predecessor \textit{upstream} of $i$ on the path ${\bar p}_{s}$.  More precisely, suppose $p_{s,k}=i$ for some $k$.  Then $u(i,s)=p_{s,k-1}$.  Let $U(i)=\{u(i,s): s\in{\cal S}\}$ denote the set of all one-hop predecessor nodes upstream of node $i$.  Furthermore, at each relay $i\notin{\cal S}$ define the reference node with respect to $i$ as the highest ranked one-hop predecessor upstream to node $i$ over all paths.  That is, ${\rm ref}(i):=\arg\displaystyle\max_{j\in U(i)}{\rm rank}(j)$. 

\begin{example}
	Set $M:=5$ and paths ${\bar p}_{1}:=(1,2,3,4,5)$ and ${\bar p}_{5}:=(5,4,3,2,1)$. Define the rank assignment ${\bar r}:=(5,3,2,1,4)$.  Then $u(2,1)=1$, $u(2,5)=3$ and $U(2)=\{1,3\}$.  It follows that ${\rm ref}(2)=1$.  Similarly, $u(3,1)=2$, $u(3,5)=4$, and $U(3)=\{2,4\}$.  It follows that ${\rm ref}(3)=2$.   
\end{example}

The valid path-rank-assignments capture all of the ways in which the tension between two opposing information flows is resolved in a decode-forward scheme when ${\bar p}:=\{{\bar p}_{1},{\bar p}_{M}\}$.  The definition of a valid path-rank-assignment thus far has been limited to the paths ${\bar p}_{1}:=(1,2,3,\ldots,M-1,M)$ and ${\bar p}_{M}:=(M,M-1,\ldots,3,2,1)$.  The symmetry of these paths simplifies the characterization of the valid path-rank-assignments.  In future work, we will define the valid path-rank-assignments over arbitrary paths.

\section{Main Result}
\label{mainresult}
Consider the TWMRC with $M$ nodes and ${\cal S}:=\{1,M\}$.  For a fixed ${\bar v}\in{\cal V}$ and any non-empty $S\subseteq{\cal S}$ define the higher orthant set (the orthant set ``above'') $A_{i}(S)$, as the set of nodes in $P_{i}(S)$ of rank higher than the reference node ${\rm ref}(i)$.  Similarly, define the lower orthant set (the orthant set ``below'') $B_{i}(S)$, as the set of nodes in $P_{i}(S)$ of rank lower than or equal to the reference node.  Finally, define the lower orthant set $\tilde{B}_{i}(S)$ as the set of nodes in $\tilde{P}_{i}(S)$ of rank lower than or equal to the reference node.  The orthant sets can be expressed as follows:
\begin{align}
	\nonumber
	A_{i}(S)&:=\{j\in P_{i}(S)\mid {\rm rank}(j)> {\rm rank}({\rm ref}(i))\},\\
	\nonumber
	B_{i}(S)&:=\{j\in P_{i}(S)\mid {\rm rank}(j)\leq{\rm rank}({\rm ref}(i))\},\\
	\nonumber
	\tilde{B}_{i}(S)&:=\{j\in \tilde{P}_{i}(S)\mid {\rm rank}(j)\leq{\rm rank}({\rm ref}(i))\}.
\end{align}
Let $L(i)$ denote the set of nodes of strictly lower rank than $i$.  That is, $L(i):=\{j\mid {\rm rank}(j)<{\rm rank}(i)\}$.  
\begin{example}
	Set $M:=5$, ${\bar p}_{1}:=(1,2,3,4,5)$, ${\bar p}_{5}:=(5,4,3,2,1)$, and ${\bar r}:=(5,3,1,2,4)$.  
	Then ${\rm ref}(4)=5$, $L(4)=\{3\}$, $A_{4}(\{1\})=\{1\}$, $A_{4}(\{5\})=\{\}$, $A_{4}(\{1,5\})=\{1\}$, $B_{4}(\{1\})=\{2,3\}$, $B_{4}(\{5\})=\{5\}$, $B_{4}(\{1,5\})=\{2,3,5\}$, $\tilde{B}_{4}(\{1\})=\{4,5\}$, $\tilde{B}_{4}(\{5\})=\{2,3,4\}$, and $\tilde{B}_{4}(\{1,5\})=\{4\}$. 
\end{example}

For any non-empty $S\subseteq{\cal S}$ and ${\bar v}\in{\cal V}$ define the constraint: 
\begin{align}
	\label{lowerbound}
	R_{S}<\displaystyle\sum_{j\in A_{i}(S)}I(X_{j};Y_{i}|X_{L(j)})+I(X_{B_{i}(S)};Y_{i}|X_{\tilde{B}_{i}(S)}).
\end{align}

Let ${\cal R}_{i}({\bar v})$ be the set of rate pairs $(R_{1},R_{M})$ that satisfy (\ref{lowerbound}) for all non-empty $S\subseteq{\cal S}$ if $i\notin{\cal S}$ and satisfy (\ref{decodeforwardupperbound}) for all non-empty $S\subseteq{\cal S}\setminus i$ if $i\in{\cal S}$.  We have the following theorem.

\begin{theorem}
	\label{maintheorem}
	For any product distribution $p(x_{1})\cdots p(x_{M})$ the following set of rate pairs is achievable:
	\begin{align}
		\nonumber
		\bigcup_{{\bar v}\in{\cal V}}\hspace{1mm}\bigcap^{M}_{i=1}{\cal R}_{i}({\bar v})
	\end{align}
\end{theorem}

\section{Proof of Theorem \ref{maintheorem}}
\label{sectionproof}
The proof is based on the block-markov, decode-forward framework in which the transmissions are divided into $B$ blocks of $T$ channel uses.  Let $R_{1}$ and $R_{M}$ denote the rates at which nodes 1 and $M$ transmit information to each other.  

\subsection{Codebook Generation}
For node 1, independently generate $2^{TR_{1}}$ i.i.d $T$-sequences ${\bar x}_{1}:=(x_{1,1},\ldots,x_{1,T})$ in ${\cal X}^{T}_{1}$ according to $p(x_{1})$.  Index them as ${\bar x}_{1}(w_{1})$, $w_{1}\in\{1,2,\ldots,2^{TR_{1}}\}$.  For node $i=2,\ldots,M-1$, independently generate $2^{T(R_{1}+R_{M})}$ i.i.d $T$-sequences ${\bar x}_{i}:=(x_{i,1},\ldots,x_{i,T})$ in ${\cal X}^{T}_{i}$ according to $p(x_{i})$.  Index them as ${\bar x}_{i}(w_{i})$, $w_{i}\in\{1,2,\ldots,2^{T(R_{1}+R_{M})}\}$.  For node $M$, independently generate $2^{TR_{M}}$ i.i.d $T$-sequences ${\bar x}_{M}:=(x_{M,1},\ldots,x_{M,T})$ in ${\cal X}^{T}_{M}$ according to $p(x_{M})$.  Index them as ${\bar x}_{M}(w_{M})$, $w_{M}\in\{1,2,\ldots,2^{TR_{M}}\}$.

\subsection{Encoding}
In each block $b\in1,\ldots,B$, nodes 1 and $M$ generate message indices $w_{1}(b)\in\{1\,\ldots,2^{TR_{1}}\}$ and $w_{M}(b)\in\{1\,\ldots,2^{TR_{M}}\}$ and transmit the $T$-sequences ${\bar x}_{1}(w_{1}(b))$ and ${\bar x}_{M}(w_{M}(b))$ respectively.  Simultaneously, each relay node $i$ chooses a message index $w_{i}(b)\in\{1,\ldots,2^{T(R_{1}+R_{M})}\}$ and transmits the $T$-sequence ${\bar x}_{i}(w_{i}(b))$.  The index $w_{i}(b)$ corresponds to a unique message pair $(w_{1}(b-d_{i,1}),w_{M}(b-d_{i,M}))$ where $d_{i,1}$ and $d_{i,M}$, referred to as the \textit{encoding delays}, are strictly positive integers.  The \textit{encoding scheme} specifies the encoding delays at the relay nodes.  At the end of each block $b$, every relay node $i$ decodes the message pair $(w_{1}(b-\tilde{d}_{i,1}),w_{M}(b-\tilde{d}_{i,M}))$ where the \textit{decoding delays}, $\tilde{d}_{i,1}$ and $\tilde{d}_{i,M}$ are strictly positive integers.  The \textit{decoding scheme} specifies the decoding delays at the relay nodes.  The relay cannot encode any message pair that it has not already decoded, so $d_{i,s}>\tilde{d}_{i,s}$.  Apart from this causality constraint, the messages decoded in one block need not determine the messages encoded in the next block; there may be many decoding schemes causally consistent with a fixed encoding scheme.    

Every valid path-rank-assignment generates an encoding scheme and a corresponding set of causally-consistent decoding schemes.  For each $s\in{\cal S}$, let $\tilde{s}:={\cal S}\setminus s$.  For any ${\bar v}\in{\cal V}$ and each $s\in{\cal S}$ and $i\in\{1,\ldots,M\}\setminus\tilde{s}$, define $f_{i,s}$ as follows: 
\begin{align}
	\label{f}
	f_{i,s}&:=
	\begin{cases}
		\displaystyle\displaystyle\sum_{k\in B_{i}(\tilde{s})}f_{k,\tilde{s}} &\text{if}\ i={\rm ref}(p_{s,i+1}),\\
		1 &\text{otherwise}
	\end{cases}
\end{align}
For each $s\in{\cal S}$ and $i\in\{2,\ldots,M-1\}$, define $d_{i,s}$ as follows:
\begin{align}
	\label{landr}
	d_{i,s}&:=\sum_{j\in P_{i}(s)}f_{j,s}.
\end{align}
\begin{example}
	Let $M=4$, ${\bar p}_{1}:=(1,2,3,4)$, ${\bar p}_{4}:=(4,3,2,1)$, and ${\bar r}:=(4,2,1,3)$.  Then $f_{1,1}=f_{4,4}+f_{3,4}+f_{2,4}$, $f_{2,1}=f_{3,4}$, $f_{3,1}=1$, $f_{4,4}=f_{2,1}+f_{3,1}$, $f_{3,4}=1$, $f_{2,4}=1$.  Expanding gives $f_{1,1}=4$, $f_{2,1}=1$, and $f_{4,4}=2$.  Therefore $d_{2,1}=f_{1,1}=4$, $d_{3,1}=f_{1,1}+f_{2,1}=5$, $d_{3,4}=f_{4,4}=2$, and $d_{2,4}=f_{4,4}+f_{3,4}=3$.
\end{example}
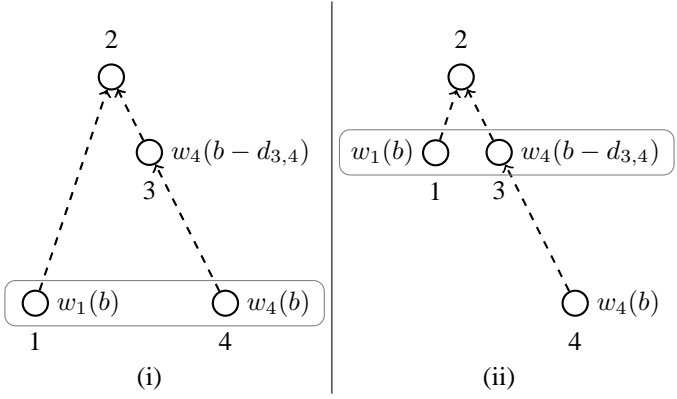
\begin{figure}
		\begin{center}
		\begin{tikzpicture}
			\node[circle,draw,thick,minimum size=0.75mm] (A) at  (0.5,0.5) {};
			\node[circle,draw,thick,minimum size=0.75mm] (B) at  (1.5,3.5) {};
			\node[circle,draw,thick,minimum size=0.75mm] (C) at  (2,2.5) {};
			\node[circle,draw,thick,minimum size=0.75mm] (D) at  (3,0.5) {};
			\draw[dashed,->] (A) -- (B);
			\draw[dashed,->] (C) -- (B);
			\draw[dashed,->] (D) -- (C);
			\node at (0.5,0) {1};
			\node at (1.20,0.5) {$w_{1}(b)$};
			\node at (3,0) {4};
			\node at (3.7,0.5) {$w_{4}(b)$};
			\node at (2,2) {3};
			\node at (3.2,2.5) {$w_{4}(b-d_{3,4})$};
			\node at (1.5,4) {2};
			\node[circle,draw,thick,minimum size=0.75mm] (F) at  (6.1,3.5) {};
			\node[circle,draw,thick,minimum size=0.75mm] (G) at  (6.6,2.5) {};
			\node[circle,draw,thick,minimum size=0.75mm] (H) at  (7.6,0.5) {};
			\draw[dashed,thick,->] (A) -- (B);
			\draw[dashed,thick,->] (C) -- (B);
			\draw[dashed,thick,->] (D) -- (C);
			\coordinate (E) at (5.1,0.5);
			\path (E) -- (F)
  			   node[circle, draw, thick,  pos=0.705, minimum size=0.75mm](I){};
			\draw[dashed,thick,->] (I) -- (F);
			\draw[dashed,thick,->] (G) -- (F);
			\draw[dashed,thick,->] (H) -- (G);
			\node at (6.6,2) {3};
			\node at (7.8,2.5) {$w_{4}(b-d_{3,4})$};
			\node at (7.6,0) {4};
			\node at (8.3,0.5) {$w_{4}(b)$};
			\node at (6.1,4) {2};
			\path (I) ++(0,-0.5) node {1};
			\path (I) ++(-0.7,0) node {$w_{1}(b)$};
			\draw (4.4,-0.7) -- (4.4,4.5);
			\node at (2,-0.5) {(i)};
			\node at (6.6,-0.5) {(ii)};
			\draw[gray,rounded corners] (0.1, 0.20) rectangle (4.3, 0.8) {};
			\draw[gray,rounded corners] (4.5,2.2) rectangle (8.9,2.8) {};
		\end{tikzpicture}
	\end{center}
	\caption{An encoding scheme and two decoding schemes that cumulatively recover ${\cal R}_{2}(\bar{p})$, where ${\bar p}_{1}=(1,2)$ and ${\bar p}_{4}:=(4,3,2)$.  In block $b$ nodes 1, 3, and 4 encode $w_{1}(b)$, $w_{4}(b-d_{3,4})$, and $w_{4}(b)$ respectively.  In decoding scheme (i), $w_{1}(b)$ and $w_{4}(b)$ are jointly decoded in block $b+d_{3,4}$.  In decoding scheme (ii) $w_{1}(b+d_{3,4})$ and $w_{4}(b)$ are jointly decoded in block $b+d_{3,4}$.}
\label{figure1}
\end{figure}

\subsection{Decoding and the Analysis of the Probability of Error}
It remains to show that there exists a set of causally-consistent decoding delay pairs $\{(\tilde{d}_{i,1},\tilde{d}_{i,M})\}$ for each $i=2,\ldots,M-1$ that allow node $i$ to recover any rate pair in ${\cal R}_{i}({\bar v})$.  First, we prove a preliminary lemma.  Define a multiple-access relay channel consisting of $M$ nodes, where node $1$ and node $M$ are sources, node $2$ is a destination, and nodes $3\ldots,M-1$ are relay nodes.  Assume block-Markovian encoding; in each block $b$ of $T$ channel uses, node $i\in\{1,M\}$ generates the message $w_{i}(b)\in\{1,\ldots,2^{TR_{i}}\}$.  Furthermore, suppose a genie reveals the message $w_{M}(b-d_{i,M})$ to relay node $i$ just before block $b$, where $d_{i,M}$ is a fixed constant for each $i$.  Hence, relay node $i$ encodes $w_{M}(b-d_{i,M})$ in block $b$.  The path vector ${\bar p}$ is defined as ${\bar p}_{1}:=(1,2)$ and ${\bar p}_{M}:=(M,M-1,\ldots,4,3,2)$.  This definition of ${\bar p}$ implies that $d_{i,M}>d_{i+1,M}$ for each relay node $i$.  This channel will be called a biased multiple-access relay channel (BMARC) since the relays only help one source and not the other.  The encoding scheme is illustrated in Figure \ref{figure1} for a BMARC of size $M=4$.  Let ${\cal R}_{2}({\bar p})$ be the set of rate pairs $(R_{1},R_{M})$ that satisfy (\ref{decodeforwardupperbound}) for all non-empty $S\subseteq{\cal S}$ for the BMARC with path ${\bar p}:=\{{\bar p}_{1},{\bar p}_{M}\}$.
\begin{lemma}
	\label{biglemma}
	Given any product distribution $p(x_{1})\cdots p(x_{M})$ the rate pairs in ${\cal R}_{2}({\bar p})$ are achievable for the BMARC.
\end{lemma}
\begin{proof}
	First we show that node 2 can decode $(w_{1}(b+d_{i,M}),w_{M}(b))$ in block $b+d_{3,M}$ if $(R_{1},R_{M})$ satisfies:
	\begin{align}
		\label{R1}
		R_{1}&<I(X_{1};Y_{2}|X_{2}\ldots X_{i})\\
		\label{RM}
		R_{M}&<I(X_{3}\ldots X_{i-1};Y_{2}|X_{2})\\
		\nonumber
		&\hspace{3mm}+I(X_{i}\ldots X_{M};Y_{2}|X_{1}X_{2}X_{3}\ldots X_{i-1})\\
		\label{R1RM}
		R_{1}+R_{M}&<I(X_{-2};Y_{2}|X_{2}),
	\end{align}
	where $X_{-2}:=\{X_{j}: 1\leq j\leq M, j\neq2\}$.  Consider the decoding of $w_{1}(b+d_{i,M})$.  Since information flows from $M$ to node 2, and $w_{M}(b)$ is transmitted by node $i$ in block $b+d_{i,M}$, nodes $2,\ldots,i$ transmit messages already decoded by node 2 and nodes $i+1,\ldots,M$ transmit messages that are new.  Hence, the mutual information in (\ref{R1}) is conditioned on nodes $X_{2},\ldots,X_{i}$.  Next consider the decoding of $w_{M}(b)$.  The messages transmitted by node 1 and nodes $i,\ldots,M$ during the blocks in which $w_{M}(b)$ is transmitted by nodes $3,\ldots,i-1$, have not been decoded by node 2.  Hence, the first mutual information in (\ref{R1}) is conditioned only on $X_{2}$.  On the other hand, the messages transmitted by nodes 1 and nodes $2,\ldots,i-1$ have been decoded by node 2 during the blocks in which nodes $i,\ldots,M$ transmit $w_{M}(b)$.  Hence, the second mutual information in (\ref{RM}) is conditioned on $X_{1},\ldots,X_{i-1}$.  It follows that (\ref{R1})-(\ref{R1RM}) is achievable.
 
	Next we show that for every $(R_{1},R_{M})$ in ${\cal R}_{2}({\bar p})$ there is some $i$, $3\leq i\leq M$ such that $(R_{1},R_{M})$ satisfies (\ref{R1})-(\ref{R1RM}).  The proof is by induction.  Consider $i=M$.  If $(R_{1},R_{M})$ satisfies (\ref{R1})-(\ref{R1RM}), then we are done.  Suppose otherwise.  Since (\ref{R1}) and (\ref{R1RM}) are boundaries of ${\cal R}_{2}({\bar p})$ at $i=M$ it follows that $R_{M}>I(X_{3},\ldots X_{M-1};Y_{2}|X_{2})+I(X_{M};Y_{2}|X_{1}\ldots X_{M-1})$.  But this together with (\ref{R1RM}) implies that $R_{1}<I(X_{1};Y_{2}|X_{2}\ldots X_{M-1})$ which satisfies (\ref{R1}) for $i=M-1$.

	Next, consider any $3<i<M$.  If $(R_{1},R_{M})$ satisfies (\ref{R1})-(\ref{R1RM}), then we are done.  Suppose otherwise.  By the inductive hypothesis, $R_{1}$ satisfies (\ref{R1}).  Furthermore (\ref{R1RM}) is a boundary of ${\cal R}_{2}({\bar p})$.  It follows that $R_{M}>I(X_{3}\ldots X_{i-1};Y_{2}|X_{2})+I(X_{i}\ldots X_{M};Y_{2}|X_{1}\ldots X_{i-1})$.  But this together with (\ref{R1RM}) implies that $R_{1}<I(X_{1};Y_{2}|X_{2}\ldots X_{i-1})$ which satisfies (\ref{R1}) for $i-1$.  The argument terminates at $i=3$ since (\ref{RM}) and (\ref{R1RM}) are boundaries of ${\cal R}_{2}({\bar p})$ at $i=3$.  Thus the lemma is proved.
\end{proof}

Figure \ref{figure1} illustrates the proof of Lemma \ref{biglemma} for $M=4$.  Finally, we show for any ${\bar v}\in{\cal V}$, the region $\bigcap^{M}_{i=1}{\cal R}_{i}({\bar v})$ is achievable.  The proof is by induction.  Consider, the two-way one-relay channel where node 1 and 3 are the sources and node 2 is the relay.  The paths are defined as ${\bar p}_{1}:=(1,2,3)$ and ${\bar p}_{3}:=(3,2,1)$.  In \cite{Xie2007} the region $\bigcap^{3}_{i=1}{\cal R}_{i}({\bar p})$ is shown to be achievable.  It is straightforward to check that there are four valid rankings for this channel, $(3,2,1)$, $(1,2,3)$, $(3,1,2)$, and $(2,1,3)$, and that the region in Theorem \ref{maintheorem} is the same as the region in \cite{Xie2007}. 

The first step of the induction is to perform a left or right extension of the two-way one-relay channel to create a two-way two-relay channel.  Without loss of generality, choose a right-side extension where node 3 becomes a relay and node 4 is added as a source and given rank 4.  The new paths are ${\bar p}_{1}:=(1,2,3,4)$ and ${\bar p}_{4}:=(4,3,2,1)$.  By the design of (\ref{f}) and (\ref{landr}), node 3 waits to receive the message from node 1 before encoding the message simultaneously transmitted by node 4.  Since node 3 knows all of the source 1 messages to its right, and all of the source 4 messages to its left, the channel it sees is the BMARC of Lemma \ref{biglemma} (or a reflection of it).  It follows from Lemma \ref{biglemma} that node 3 can recover all the rate pairs in ${\cal R}_{3}({\bar p})$, which is equivalent to the region defined by $R_{S}<I(X_{B_{3}(S)};Y_{3}|X_{\tilde{B}_{3}(S)})$ for all non-empty $S\subseteq{\cal S}$.  Note that $A_{3}(S)=\{\}$ for all non-empty $S\subseteq{\cal S}$.

Since node 2 does not know the transmissions from node 4 a priori, the previous mutual informations that describe the contributions of nodes 1 and 3 remain unchanged; they do not include $X_{4}$.  These contributions are expressed by $I(X_{B_{2}(S)};Y_{2}|X_{\tilde{B}_{2}(S)})$ for all non-empty $S\subseteq{\cal S}$.  The right-side extension forces node 2 to decode $w_{1}(b)$ before $w_{4}(b)$.  As a result, the contribution of node 4 as seen by node 2 is $I(X_{4};Y_{2}|X_{1}X_{2}X_{3})$ or equivalently, $I(X_{4};Y_{2}|X_{L(4)})$.  Hence relay node $i\in\{2,3\}$ can recover any rate pair in the region defined by $R_{S}<\displaystyle\sum_{j\in A_{i}(S)}I(X_{j};Y_{i}|X_{L(j)})+I(X_{B_{i}(S)};Y_{i}|X_{\tilde{B}_{i}(S)})$ for all non-empty $S\subseteq{\cal S}$.  To finish the first inductive step, observe that node 1 and node 4 each see a one-way multiple-relay channel, which is a simple BMARC.  Therefore it follows from Lemma \ref{biglemma} that node $i\in\{1,4\}$ can decode at any rate in ${\cal R}_{i}({\bar p})$ as defined by (\ref{decodeforwardupperbound}).  Thus the first inductive step is proved.

Now assume by induction that the Theorem is true for $M-1$.  We will show it must be true for $M$.  Without loss of generality, consider a right-side extension that changes node $M-1$ from a source into a relay node and adds a source node $M$ where ${\rm rank}(M)=M$.  The new paths are ${\bar p}_{1}:=(1,2,3,\ldots,M)$ and ${\bar p}_{M}:=(M,M-1,M-2,\ldots,2,1)$.  By the design of (\ref{f}) and (\ref{landr}), node $M-1$ waits to receive the message from node 1 before encoding the message simultaneously transmitted by node $M$.  Since node $M-1$ knows all of the source 1 messages to its right, and all of the source $M$ messages to its left, the channel it sees is the BMARC of Lemma \ref{biglemma} (or a reflection of it).  It follows from Lemma \ref{biglemma} that node $M$ can recover all the rate pairs in ${\cal R}_{M-1}({\bar p})$, which is equivalent to the region defined by $R_{S}<I(X_{B_{M-1}(S)};Y_{M-1}|X_{\tilde{B}_{M-1}(S)})$ for all non-empty $S\subseteq{\cal S}$.  Note that $A_{M-1}(S)=\{\}$ for all non-empty $S\subseteq{\cal S}$.

Since relay node $i\in\{2,\ldots,M-2\}$ does not know the transmissions from node $M$ a priori, the previous mutual informations that describe the contributions of nodes $1\ldots,M-1$ remain unchanged; they do not include $X_{M}$.  These contributions are expressed by $\displaystyle\sum_{j\in A_{i}(S)\setminus M}I(X_{j};Y_{i}|X_{L(j)})+I(X_{B_{i}(S)};Y_{i}|X_{\tilde{B}_{i}(S)})$ for all non-empty $S\subseteq{\cal S}$.  The right-side extension forces relay node $i$ to decode $w_{1}(b)$ before $w_{M}(b)$.  As a result, the contribution of node $M$ as seen by node $i$ is $I(X_{M};Y_{i}|X_{1}\ldots X_{M-1})$ or equivalently, $I(X_{M};Y_{i}|X_{L(M)})$.  Hence relay node $i\in\{2,\ldots,M-1\}$ can recover any rate pair in the region defined by $R_{S}<\displaystyle\sum_{j\in A_{i}(S)}I(X_{j};Y_{i}|X_{L(j)})+I(X_{B_{i}(S)};Y_{i}|X_{\tilde{B}_{i}(S)})$ for all non-empty $S\subseteq{\cal S}$.  To finish the inductive step, observe that node 1 and node $M$ each see a one-way multiple-relay channel which is a simple BMARC.  Therefore it follows from Lemma \ref{biglemma} that node $i\in\{1,M\}$ can decode at any rate in ${\cal R}_{i}({\bar p})$ as defined by (\ref{decodeforwardupperbound}).  

The definition of a valid path-rank-assignment guarantees that we can always start with the three nodes of lowest rank and reach the general two-way $M$-relay channel by a sequence of left and right extensions with nodes of successively higher rank.  Thus Theorem \ref{maintheorem} is proved.  

\section{Concluding Remarks}
\label{conclusion}
We showed that the rate regions of all interesting decode-forward schemes are different realizations of a single expression that depends on a rank assignment.  This discovery makes it possible to characterize the complete achievable rate region.  It remains to be seen whether some version of Theorem \ref{maintheorem} is also true for multi-source, multi-relay, all-way channels.

\section{Acknowledgments}
The authors would like to thank Xiugang Wu for helpful discussions on random binning and network coding.  This material is based upon work partially supported by NSF Contract CNS-1302182,  AFOSR Contract FA9550-13-1-0008, and NSF Science \& Technology Center Grant CCF-0939370.

\bibliographystyle{IEEEtran}
\bibliography{ResubmitISIT2016}

\begin{thebibliography}{1}
\providecommand{\url}[1]{#1}
\csname url@samestyle\endcsname
\providecommand{\newblock}{\relax}
\providecommand{\bibinfo}[2]{#2}
\providecommand{\BIBentrySTDinterwordspacing}{\spaceskip=0pt\relax}
\providecommand{\BIBentryALTinterwordstretchfactor}{4}
\providecommand{\BIBentryALTinterwordspacing}{\spaceskip=\fontdimen2\font plus
\BIBentryALTinterwordstretchfactor\fontdimen3\font minus
  \fontdimen4\font\relax}
\providecommand{\BIBforeignlanguage}[2]{{%
\expandafter\ifx\csname l@#1\endcsname\relax
\typeout{** WARNING: IEEEtran.bst: No hyphenation pattern has been}%
\typeout{** loaded for the language `#1'. Using the pattern for}%
\typeout{** the default language instead.}%
\else
\language=\csname l@#1\endcsname
\fi
#2}}
\providecommand{\BIBdecl}{\relax}
\BIBdecl

\bibitem{Shannon}
C.~E. Shannon, ``Two-way communication channels,'' in \emph{In Proc. 4th
  Berkeley Symp. Math, Statist. Probab}, 1961, pp. 611--644.

\bibitem{ElGamal}
T.~Cover and A.~El~Gamal, ``Capacity theorems for the relay channel,''
  \emph{IEEE Trans. Inf. Theory}, vol.~25, no.~5, pp. 572--584, Sep 1979.

\bibitem{Xie2005}
L.-L. Xie and P.~Kumar, ``An achievable rate for the multiple-level relay
  channel,'' \emph{IEEE Trans. Inf. Theory}, vol.~51, no.~4, pp. 1348--1358,
  April 2005.

\bibitem{Sankar}
L.~Sankar, G.~Kramer, and N.~B. Mandayam, ``Offset encoding for multiple-access
  relay channels,'' \emph{IEEE Trans. Inf. Theory}, vol.~53, no.~10, pp.
  3814--3821, Oct 2007.

\bibitem{Xie2007}
L.-L. Xie, ``Network coding and random binning for multi-user channels,'' in
  \emph{10th Canadian Workshop on Information Theory}, June 2007, pp. 85--88.

\bibitem{Ponniah}
J.~Ponniah and L.-L. Xie, ``An achievable rate region for the two-way two-relay
  channel,'' in \emph{ISIT}, July 2008, pp. 489--493.

\bibitem{NoisyNetworkCoding}
S.~Lim, Y.-H. Kim, A.~El~Gamal, and S.-Y. Chung, ``Noisy network coding,''
  \emph{IEEE Trans. Inf. Theory}, vol.~57, no.~5, pp. 3132--3152, May 2011.

\bibitem{WuXie}
X.~Wu and L.-L. Xie, ``On the optimal compressions in the compress-and-forward
  relay schemes,'' \emph{IEEE Trans. Inf. Theory}, vol.~59, no.~5, pp.
  2613--2628, May 2013.

\bibitem{Schein}
B.~Schein, ``Distributed coordination in network information theory,'' Ph.D.
  dissertation, Massachusetts Institute of Technology, 2001.

\end{thebibliography}

\end{document}